\documentclass[11pt,letterpaper]{article}

\usepackage{geometry}

\usepackage{latexsym}
\usepackage{epsfig}
\usepackage{enumerate,amsmath,amsthm,latexsym,amssymb,hyperref}
\usepackage{bbm}
\usepackage{enumerate}
\usepackage{float}
\usepackage{cite}

\usepackage[usenames]{xcolor}

\usepackage{caption}
\usepackage{subcaption}

\usepackage{cleveref}

\crefname{theorem}{Theorem}{Theorems}
\crefname{lemma}{Lemma}{Lemmas}
\crefname{equation}{}{}

\usepackage{graphicx}

\def\cT{{\mathcal T}}
\def\blue{\color{blue}}

\def\red{\color{red}}
\def\violet{\color{violet}}

\newcommand{\brac}[1]{\left( #1 \right)}
\newcommand\bfrac[2]{\left(\frac{#1}{#2}\right)}
\def\E{{\bf E}}

\def\Pr{\mathbb{P}}
\def\rai{\rightarrow \infty}

\newcommand{\ind}[1]{\,\mathbbm{1}_{\{{#1}\}}}
\newcommand{\ignore}[1]{ }

\def\d{\delta}

\def\eps{\epsilon}

\def\Th{\Theta}
\def\l{\lambda}

\def\om{\omega}
\def\Om{\Omega}

\def\etal{{\it{et al.}}}

\newcommand{\ra}{\rightarrow}

\newcommand{\imp}{\implies}

\def\sm{\! \setminus \!}
\newcounter{rot}

\newcommand{\floor}[1]{\lfloor #1 \rfloor}
\newcommand{\ceil}[1]{\lceil #1 \rceil}

\newcommand{\TheoremExpander}

\newtheorem{theorem}{Theorem}
\newtheorem{lemma}[theorem]{Lemma}
\newtheorem{corollary}[theorem]{Corollary}
\newtheorem{Remark}{Remark}

\newlength{\wdth}

\begin{document}

%
\title{Discrete Incremental  Voting on Expanders}

\author{Colin Cooper\thanks{
Department of  Informatics,
King's College London, London WC2R 2LS, UK.
Research supported  at the University of Hamburg, by a  Mercator fellowship from DFG  Project 491453517}
\and Tomasz Radzik\thanks{Department of  Informatics,
King's College London, London WC2R 2LS, UK.}
\and Takeharu Shiraga\thanks{
Department of Information and System Engineering, Chuo University, Tokyo, Japan.
Research supported by JSPS KAKENHI Grant Number 19K20214.}
}

\maketitle \makeatother

\begin{abstract}
Pull voting is a random process in which vertices of a connected graph have initial opinions chosen from a set of $k$ distinct opinions, and at each step a random vertex alters its opinion to that of a randomly chosen neighbour.
If the system reaches a state where each vertex holds the same opinion, then this opinion will persist forthwith.

In general the opinions  are regarded as incommensurate,
whereas in 
this paper 
we consider a type of pull voting suitable for integer opinions such as  $\{1,2,\ldots,k\}$ which can be compared on a linear scale; for example,
1 ('disagree strongly'), 2 ('disagree'), $\ldots,$ 5 ('agree strongly'). On observing the opinion of a random neighbour,  a vertex updates its opinion by a discrete change towards the value of the neighbour's opinion, if different. 

Discrete incremental voting is a pull voting process which mimics this behaviour.
At each step a random vertex alters its opinion towards that of a randomly chosen neighbour;
increasing  its opinion by $+1$ if the opinion of the chosen neighbour is larger, or decreasing its opinion by $-1$, if the opinion of the neighbour is smaller. 
If initially there are only two adjacent integer opinions, for example $\{0,1\}$, incremental voting coincides with  pull voting, but if initially there are more than two opinions
this is not the case.

For an $n$-vertex graph $G=(V,E)$, let $\l$ be the absolute second eigenvalue of the transition matrix $P$ of a simple random walk on $G$.
Let the  initial opinions of the vertices be chosen from  $\{1,2,\ldots,k\}$.
Let $c=\sum_{v \in V} \pi_v X_v$, where $X_v$ is the initial opinion of vertex $v$, and $\pi_v$ is the stationary distribution of the vertex.
Then provided $\l k=o(1)$ and $k=o(n/\log n)$, with high probability (whp)
the final opinion is the initial weighted average  $c$ suitably rounded to 
$\floor{c}$ or $\ceil{c}$.

If $G$ is a regular graph then $c$ is  the average opinion, and
with high probability the final opinion held by all vertices is $\floor{c}$ or $\ceil{c}$.
\end{abstract}

\maketitle

\section{Introduction}

\paragraph{Background on distributed pull voting.}
Distributed voting has  applications in various
fields of computer science including consensus and leader election in large networks
\cite{BMPS04,HassinPeleg-InfComp2001}.
Initially, each vertex has some value chosen from a set $S$, and the aim is
that the vertices reach consensus on
(converge to) the same value, which
should, in some sense, reflect the
initial distribution of the values.
Voting algorithms are usually simple,
fault-tolerant, and easy to implement \cite{HassinPeleg-InfComp2001,Joh89}.

Pull voting is a  simple form of distributed voting.
At each step, in the asynchronous process a randomly chosen vertex 
replaces its opinion with that of  randomly chosen neighbour.
The probability a particular opinion, say opinion $A$, wins
is $d(A)/2m$, where $d(A)$ is the sum of the degrees of the vertices initially holding opinion $A$, and
$m$ is the number of edges in the graph; see Hassin and Peleg~\cite{HassinPeleg-InfComp2001}
and Nakata~\etal~\cite{Nakata_etal_1999}.

The pull voting process has been modified to consider two or more opinions at each step. 
The aim of this modification is twofold; to ensure the majority (or plurality) wins, and to speed up the run time of the process.
Work on {\em best-of-$k$} models, where
a vertex replaces its opinion with the opinion most represented in a sample of 
$k$ neighbouring vertices,
includes
\cite{AD,Becchetti, Bec2, Bec3, Petra,CER,CRRS,Ghaffari,NanNico,NS}.

\paragraph{Discrete incremental voting: An introduction.}
We assume the initial opinions of the vertices are chosen from among the integers
$\{1,2,...,k\}$.
As a simple example, suppose the entries   reflect the views of the vertices about some issue, and range from  1 ('disagree strongly') to $k$ ('agree strongly'). It seems unrealistic that a vertex would completely change its opinion to that of a neighbour (as in pull voting) based only on observing what the neighbour thinks.

However, people being what they are, 
it seems possible that 
on observing what a neighbour thinks, 
they may modify their own opinion slightly towards the neighbour's opinion.
In the simplest case,  suppose that the selected vertex $v$ has opinion $i$ and observes a
neighbour $w$ with opinion~$j\neq i$.
Vertex $v$ then changes its opinion by $1$ towards the opinion held by vertex $w$.
That is, if $j>i$, then
vertex $v$ modifies its opinion to $i+1$. Similarly, if the observed neighbour $w$ has opinion $j<i$,
vertex $v$ modifies its opinion to $i-1$.
The neighbour $w$ does not change its opinion
at this interaction.
The convergence aspects of this process, 
including the value the process converges to and the time of convergence, 
are the topics of this paper.

\paragraph{Asynchronous  discrete incremental voting: Definition of process.}
Let
$G = (V,E)$ be a connected graph with $n$ vertices
and $m$ edges, and
let $X=(X_v: v\in V)$ be a vector of integer opinions.
At a given step, a vertex $v$ and a neighbour $w$ of $v$ are chosen according to some rule.
The value  $X_v$ of the chosen vertex is updated to $X'_v$ as follows, 
\begin{equation}\label{Tab}
\left.
\begin{array}{lll}
X_v <X_w & \imp & X'_v=X_v+1\\
X_v=X_w & \imp & X'_v=X_v\\
X_v > X_w & \imp & X'_v=X_v-1
\end{array}
\; \right. 
\end{equation}
For all other vertices $u \ne v$ (including vertex $w$),  $X'_u=X_u$.

We  consider two  asynchronous  processes which differ in the way the vertex $v$ and its neighbour $w$ are chosen;
an event which we call ``$v$ chooses $w$.''

{\sc Asynchronous vertex process}: \; Choose a random neighbour of a random vertex.
A~vertex $v \in V$ is chosen with probability (w.p.) $1/n$ and a neighbour $w$ of $v$ is chosen w.p.\ $1/d(v)$,
where $d(v)$ is the degree of vertex $v$.
Thus 
\begin{equation}\label{avtrans}
\Pr(v \text{ chooses }w)= \frac 1n  \frac{\ind{(v,w) \in E(G)}}{d(v)},
\end{equation}

{\sc Asynchronous edge process}: \; Choose a random endpoint of a random edge.
Edge $e \in E(G)$ is chosen w.p.\ $1/m$, 
and one of its endpoints is chosen w.p.\ $1/2$ as vertex $v$, while the other
endpoint is the selected neighbour vertex $w$.
Thus 
\[
\Pr(v \text{ chooses  }w)= \frac 1{2m}{\ind{(v,w) \in E(G)}}.
\]
The edge process can be seen as a vertex process where  $v$ is sampled with probability $\pi_v=d(v)/2m$ rather than uniformly at random.

In order to reach a consensus opinion in discrete incremental voting, all other opinions must be eliminated. 
The only way to irreversibly 
eliminate an opinion,
is to remove one of the two extreme opinions in the order.
This being repeated until a single opinion remains.
Returning to our original example, 1 ('disagree strongly'), 2 ('disagree'), 3 ('indifferent'), 4 ('agree'), 5 ('agree strongly'),
suppose we start with each vertex having one of the opinions in $\{1,2,5\}$.  
Then a possible evolution of the system 
(that is, evolution of the set of opinions present in the system) is
\[
\{1,2,5\}\! \to\! \{1,2,4\}\! \to\! \{1,2,3,4\}\! \to\! \{2,3,4\}\! \to\! \{2,4\}\! \to\! \{2,3\}\! \to\! \{3\},
\]
where
{the set of opinions at the beginning of each of these stages is indicated, and}
'$\ra$' represents a sequence of one or more steps constituting a stage.
Intermediate values may disappear 
and then appear again; in the above example,
opinion $3$ disappears in stage four and  appears again in stage five.
Eventually, as extreme values disappear, 
we reach a stage  when  only two adjacent values  remain.
In the example above, there are only opinions $2$ and $3$ during the last stage. 
At this point the process reverts to ordinary two-opinion pull voting, and finally 
one of these two opinions wins (in this example, opinion 3 wins).

On a connected graph, discrete incremental voting is a Markov process in which the set of possible opinions  decreases when extreme values
are removed. The final singleton states are absorbing, all other states being transient. 
In the example above, the process reaches the absorbing state $\{3\}$.
The extreme values in order of removal are $5,1,4,2$.

\paragraph{The main features of discrete incremental voting.}

The general model of \emph{pull voting} on connected graphs regards the opinions as incommensurate, and thus not comparable on a numeric scale. 
As mentioned earlier, 
the final opinion is chosen with probability proportional to 
the sum of the degrees of the vertices initially holding this opinion. 
Thus, for regular graphs, the
most likely value is \emph{the mode} of the initial opinions.
In contrast to this, \emph{median voting}, introduced and analysed in Doerr~\etal~\cite{Doerr}, 
considers opinions drawn from an ordered set, 
and aims to converge to \emph{the median} value of the initial opinions. 
At each step a random vertex selects two neighbours and replaces its opinion by the median of all three values (including its own current value).
On the complete graph, if $l$ denotes the index of the final consensus value 
(that is, the process converges to 
the $l$-th smallest of the
initial values), then w.h.p.\ $l\in [n/2 - O(\sqrt{n \log n}), n/2 + O(\sqrt{n \log n})]$.

The variant of pull voting considered in this paper,
 {\em discrete incremental voting} (DIV), 
regards opinions as integers in the range $\{1,\ldots,k\}$ and converges, under suitable conditions, 
to \emph{the average} opinion of the group.
Seen in this context,  pull voting, median voting
and our discrete incremental voting,  mirror (respectively) the statistical measures of Mode, Median and Mean.

A concrete  application of discrete incremental voting is to find the integer average of  integer  weights held at the vertices of  
 a network. The DIV process does this using only the pull voting paradigm, a very weak type of interaction when 
only one of the two interacting vertices updates its state.
Suppose
that the initial  average\footnote{The type of average returned depends on the algorithm. The edge process returns a simple average while
the vertex process returns a degree weighted average.} 
of the weights is $c$. 
We prove that for many classes of expanders, with high probability\footnote{
With high probability (w.h.p.): with probability $1-o_n(1)$.} the final unique opinion in the incremental voting process is $c$, 
if $c$ is integer, and either $\floor{c}$ or $\ceil{c}$ otherwise.
The accuracy of the outcome is perhaps surprising, as pull voting achieves the mode only with positive probability, and 
median voting with $k$ distinct opinions can err from the true median by up to $k/\sqrt n$ values. 

Typical  asynchronous averaging algorithms operating with integer values are inspired by load balancing. 
The simplest example of this is in population protocol models, where a random edge is chosen and
its endpoints update
their values (loads)
to the round up and round down of the half of the total load over the edge.
If the loads at the vertices are $a$ and $b$, then the new loads are $\floor{(a+b)/2},\, \ceil{(a+b)/2}$.
 Unless the initial average $c$ is integer,
this process can lead at best to  a mixture of $\floor{c}$ and $\ceil{c}$ values at the vertices.
See~\cite{Petra2} for details, where
it was shown  that with high probability
this load balancing process reaches a state consisting of three consecutive values around the initial average within $O(n \log n + n\log k)$ steps, for the initial loads in $\{1,\ldots,k\}$.

The approach of averaging along an edge requires a simultaneous coordinated update of the edge endpoints.
Incremental voting is simpler;  in each step only one vertex changes its state. 
On many expander graphs, all  opinions are quickly replaced by the 
the initial integer averages $\floor{c}$ and $\ceil{c}$. 
The process then enters the final stage, which corresponds to classic two-value pull voting
which  is equivalent to probabilistic rounding. 
On completion, all vertices hold the same integer  opinion,  e.g.,  $\floor{c}$. 

Thus, although DIV does not conserve the total initial weight of the opinions,  it   
converges   w.h.p.\ to the  initial integer average in the edge process (or to the degree weighted average in the vertex process). 
The advantage of DIV 
comes from the extreme opinions contracting rapidly towards the initial average. This avoids the need for 
more complex interactions such as load balancing across edges.

 \paragraph{Two-opinion pull voting: The last stage in  incremental voting.}
In randomized pull voting, a vertex updates its value at a given step by choosing (pulling) the value of a 
neighbour chosen uniformly at random. 
Two-opinion pull voting is a special case of pull voting where initially there are only two opinions, usually written $\{0,1\}$.

The simplest  case in which incremental voting differs from pull voting is where the  opinions are chosen from  three adjacent integers, say $\{0,1,2\}$ or $\{1,2,3\}$.
In general we assume the initial values are in the range  $\{1, ..., k\}$.
In order for discrete incremental voting to ever finish, it must
reach a stage where only two adjacent opinion  values $\{i,i+1\}$ remain; at which point it reverts
to two-opinion pull voting.

For future reference, we  note the probability of winning in two-opinion pull voting starting with values $\{0,1\}$.
Let $A_j$ be the set of vertices with value  $j \in\{0,1\}$, and
$N_j=|A_j|$, where $N_0 + N_{1}=n$. Let $d(A)=\sum_{v \in A} d(v)$ be the total degree of set $A$.  The probability
that $i$ wins is
\begin{equation} \label{probw}
\Pr(i \text{ wins})=\frac{N_i}{n} \quad (\text{Edge process}), \qquad
{\Pr(i \text{ wins})=\frac{d(A_i)}{2m} \quad(\text{Vertex process}).}
\end{equation}

\paragraph{Discrete incremental voting. Results.}

This section gives the main result.
Everything depends on proving Theorem \ref{thm:expander_gen}, which shows that on expanders the range of opinions is 
reduced relatively quickly from $k$ to two adjacent ones. The  outcome of the final stage of two-opinion voting, 
Theorem \ref{thm:expander1}, then follows from Lemma \ref{ThA}.

Let $G=(V,E)$ be a connected graph with $n$ vertices and $m$ edges. Let $P$ be the transition matrix of a simple random walk on $G$,
defined by
$P(v,u)=\mathbbm{1}_{\{v,u\}\in E}/d(v)$ for $v,u\in V$. 
Assume $P$ is aperiodic, so that  the stationary distribution of vertex $v$ is $\pi_v=d(v)/2m$.
Let $\pi_{\min}=\min_{v\in V}\pi_v$,
and let $\lambda$ be the second largest eigenvalue in absolute value of the transition matrix $P$.
The set $\{1,2,...,k\}$ is denoted by $[k]$. With high probability (w.h.p.) means with probability $1-o_n(1)$.

\renewcommand{\TheoremExpander}{%
{\sc Reduction to two adjacent opinions.}\\
Consider asynchronous  incremental voting on  $G$ with opinions from $[k]$.
Suppose
$\lambda k =o(1)$, $k=o(n/\log n)$
and $\pi_{\min}=\Theta(1/n)$.
Then,  in the vertex process within $T=o(n^2)$
steps with high probability~only two consecutive opinions $i$ and $i+1$ remain.
}

\begin{theorem}
\label{Expandy}\label{thm:expander_gen}
\TheoremExpander
\end{theorem}

The expected upper bound on $T$ in Theorem \ref{thm:expander_gen} is 
\begin{equation}\label{T-val}
\E[T]=O(k\cdot  n\log n +n^{5/3} \log n + \lambda k \cdot n^2 +\sqrt{\lambda} \cdot n^2 ).
\end{equation}

Let $Z(t)=n \sum_{v \in V}\pi_v X_v(t)$ be the degree biased total weight.
Let $Z(0)=cn$ the initial  total weight, so that $c$ is the initial weighted average.
Combining \cref{thm:expander_gen} and (iii) of \cref{ThA} implies the following theorem.
\begin{theorem}\label{thm:expander1}
{\sc Asynchronous incremental voting on expanders.}\\
    Consider  asynchronous incremental voting in the  vertex process on a graph $G$ with opinions from $[k]$.
    Suppose $\lambda k=o(1)$, $k=o(n/\log n)$, and $\pi_{\min}=\Theta(1/n)$.
    Let $i$ such that $i\leq c<i+1$, where $c$ is the initial weighted average. Then w.h.p. the winning opinion is $i=\floor{c}$ 
    with probability $p\sim i+1-c$, and is $i+1=\ceil{c}$ with probability $q\sim c-i$.
\end{theorem}

For near regular graphs, Theorems \ref{thm:expander_gen} and \ref{thm:expander1} hold equally for the edge process by the following remark.
\begin{Remark}\label{St=Zt}
For regular graphs, $Z(t)$ coincides with
$S(t)=\sum_{v \in V} X_v(t)$, the total sum of all opinions held at step $t$. For graphs where for all vertices $v$,
$d(v)=(1+o(1))d$,  we have $\pi_v \sim 1/n$. Thus $S(t)=(1+o(1))Z(t)$, and results for the 
asynchronous vertex process also hold for the asynchronous edge process.
 \end{Remark}

\paragraph{Previous work.}
Discrete incremental voting was analysed  in~\cite{DIVFULL} using ad-hoc methods for
the complete graph $K_n$, and random graphs $G(n,p)$ with restrictions on the range of $p$. 
Unlike ordinary pull voting, no general method is known to predict the outcome of incremental voting.
Indeed it was also shown in ~\cite{DIVFULL} that there exist 
examples of graphs where an opinion other than $\lceil c \rceil$ or $\lfloor c\rfloor$ can win with constant probability
in the case where $\lambda k=\Omega(1)$.
The $n$-vertex path graph has $\lambda=1-O(1/n^2)$. Starting with initial opinions \{0,1,2\}, there exist initial 
configurations of opinions such that each of the three opinions  wins with positive probability, (see Theorem 3 of ~\cite{DIVFULL}) .

\paragraph{Graphs with small second eigenvalue.} 
To illustrate the applicability of Theorem \ref{thm:expander1}, we give three examples of classes of graphs with
small second eigenvalue. 
For further examples such as near regular graphs see e.g., \cite{CERReigen} for eigenvalue bounds.

\begin{itemize}
\item  The second eigenvalue of the complete graph $K_n$ has absolute value $\l=1/(n-1)$.
Thus  Theorem \ref{thm:expander1} holds  for $K_n$ provided $k=o(n/\log n)$. 

\item Random $d$-regular graphs with $n$ vertices, where $d\leq n/2$.
Then $\lambda \leq O(1/\sqrt{d})$ with high probability~(see \cite{CGJ18,TY19}).
Hence, we can apply \cref{thm:expander1}~if $k=o(\sqrt{d})$. 

\item Random graphs $G(n,p)$, where $2(1+o(1)) \log n \le np \le 0.99 n$.
Then $\lambda \le (1+o(1)){2}/{\sqrt{np}}$ w.h.p.~(see \cite{ACO}, Theorem 1.2).
As $\pi_{\min}=\Theta(1/n)$ w.h.p.~here, we can apply \cref{thm:expander1}~provided $k = o(\sqrt{np})$.
\end{itemize}

\paragraph{Strong concentration of final average.}  
For graphs with sufficiently small second eigenvalue, the final integer average obtained from DIV can exhibit remarkable concentration.  
It follows from \eqref{eq:exponential decay} with $h=\om \log n$, that 
 the probability $T$ exceeds $T^*=\om \log n \;\E[T]$ is at most $O(n^{-\om})$. 
Thus w.h.p. by step $T^*$ only two consecutive opinions remain.

As an example, consider $K_n$. In which case $\l= 1/(n-1)$,  and
\[
\E [T]=
O(k\cdot  n\log n +n^{5/3} \log n + k \cdot n+ n^{3/2})=O(k\cdot  n\log n +n^{5/3} \log n).
\]
 Assume $\min(c-\floor{c},\ceil{c}-c) \ge \d$ constant.
Suppose the process fails to return the original integer average (either $\floor{c}$ or$\ceil{c}$). 
Then  by step $T^*$, the total weight $W(t)$ must have changed by at least $\d n$. In which case
using \eqref{W-conc}, 
\[
\Pr\left[|W(t)-W(0)|\geq \d n \right]\leq 2e^{-\frac{\d^2 n^2}{O(T^*)}}= \exp\brac{-\Om \bfrac{\d^2 n}{\om \log^2 n(k+n^{2/3})}}.\]
For example, if $k=O(n^{2/3})$ the RHS above is  $ e^{ -\Om(n^{1/4})}$. In which case the probability 
DIV does not return $\floor{c}$ or$\ceil{c}$ as claimed 
is $O(n^{-\om})$.

\paragraph{Proof outline.}
\ignore{
Let $T_2(G)$ be the worst-case\footnote{By worst case we mean with the worst initial mixture of 0's and 1's arranged on the vertices in
the worst manner.} 
expected time to consensus for two-opinion pull voting on a connected graph $G$.
For the complete graph $K_n$  
and $G_{n,p}$ (for a sufficiently large $p$) we have
 $T_2=\Th(n^2)$, by comparison with a coalescing random walk. 

It is shown in  Lemma~\ref{T2} below, that the expected time
for an extreme opinion to disappear  in incremental voting is $O(T_2)$. 
Thus the expected time to consensus on $G$ with initial
opinions in $\{1,\ldots,k\}$ is at most $O(k T_2)$. However this is too slow to preserve the initial average. 
}

In Theorem \ref{Expandy}
the number of opinions is  reduced from $k$ to three consecutive integer values  in $O(nk \log n)+O(n^2 \lambda k)$
expected steps, and in a further $O(n^{5/3} \log n) +O(n^2 \sqrt{\lambda})$ expected steps only two consecutive opinions remain. 
Thus if $T$ is the number of steps  needed to reduce the system to two consecutive opinions, then as in \eqref{T-val},
\[
\E(T)=O(nk \log n +n^2 \lambda k +n^{5/3} \log n +n^2 \sqrt{\lambda}).
\]
Let $S(t)$ be the total sum of all opinions held at step $t$, and  $S(0)=cn$ the initial total. 
By Azuma's lemma, the total $S(t)$ remains concentrated w.h.p. at any step $t=o(n^2)$.
Thus provided $T=o(n^2)$
the final two-opinion voting stage begins with $S(T) \sim S(0)$.

The actual method to reduce the number of opinions depends on observing that the update probability
\eqref{avtrans} is proportional to the transition probability of a random walk, 
and then applying the expander mixing lemma (Lemma \ref{lem:EML})
and a linear voting lemma (Lemma \ref{lem:pullupper}).
This has to be done with some care because, in order to maintain concentration of the average opinion around its initial value,
we will need to prove that all but two consecutive opinions disappear
within $o(n^2)$ steps. 

A stage of the proof consists of  removing one of the extreme opinions (either the  smallest
or the largest opinion), thus reducing the range of the opinions by one. 
In Section \ref{reduce} we initialize each stage of the
proof by using Lemma \ref{lem:extremal_expander}, (an application of the expander mixing lemma, see Lemma \ref{lem:EML}),
to decrease the stationary measure of one of the
extreme opinions to a  threshold $\epsilon$ with probability
$1/2$.  In the case where at least 4 opinions remain, this occurs within $O(n\log (1/\epsilon))$ steps,
provided $\epsilon$ is at least $\lambda^2$. 

It is known from a proof in  \cite{CR16} on the  voting model, presented in Section \ref{couple} of this paper  as Lemma \ref{lem:pullupper}  (and adapted to DIV subsequently in that section),
that if the stationary measure of one opinion is
sufficiently small at the beginning (our value $\epsilon$), it disappears within the time
$T(\epsilon)$ specified in Lemma \ref{lem:pullupper} with probability $1/2$. 
However,  with positive probability $1/2$ neither extreme opinion  disappears,   and we are
back  where we started.
This requires a back and forth argument between the expander mixing approach, Lemma \ref{lem:extremal_expander},
and Lemma \ref{lem:pullupper}, with a probability of success $1/4$ at each try until we
succeed. In Section \ref{sec:Proof of main thm} this is all put together to prove Theorem  \ref{Expandy}, the
times to remove an extreme opinion with probability $1/4$ being summarised in
\eqref{Tvals}, thus leading to the expected time to finish $\E T$ 
stated in \eqref{T-val} below Theorem \ref{Expandy}.

Interestingly none of these arguments  sheds any light on the final
consensus opinion, or tells us what the final two consecutive opinions are. This is
obtained indirectly in Lemma \ref{ThA} via a separate martingale argument. Namely, that if in
$o(n^2) $ steps only two consecutive opinions are left, then the total weight of
the opinions remains concentrated around its original value, hence leading us back
to the start of the proof explanation several paragraphs above.

\paragraph{Notation.}
For functions $a=a(n)$ and $b=b(n)$, $a \sim b$  denotes $a=b(1+o(1))$,
where $o(1)=o_n(1)$ is a function of $n$ which tends to zero as $n \rai$. We use $\om$ to denote a generic quantity tending to infinity as $n \rai$, but suitably slowly as required in the given proof context. 
An event $A$ on an $n$-vertex graph holds with high probability (w.h.p.),
if $\Pr(A)=1-o_n(1)$. We also use the notation
$\|\pi\|_2 = \sqrt{\sum_{v\in V} \pi_v^2}$, $\|\pi\|_\infty = \max_{v\in V}\pi_{v}$, and $\pi_{\min}=\min_{v\in V}\pi_v$, where $\pi_v=d(v)/2m$.

\section{Basic properties of  incremental voting }\label{Sec2}

With the exception of Lemma \ref{lem:Azuma-Hoeffding}, the following results in this section are from  \cite{DIVFULL}. We restate them for convenience.

{Let $X(t)=(X_v(t): v\in V)$ be the vector of integer opinions held by the vertices at step~$t$,
where $X(0)$ is the vector of initial opinions.
We use the notation $A_i(t) = \{v\in V: X_v(t) =i\}$ for the set of vertices holding opinion $i\in \{1,...,k\}$ at time~$t$.
The weight of  vertex  $v$ in the edge process is $X_v$, and the weight in the vertex process is $\pi_v X_v$, where $\pi_v=d(v)/2m$ and $m$ is the number of edges of the graph.

Let $S(t)$ be the total weight of the edge process at step $t \ge 0$:
$S(t)=\sum_{v\in V} X_v(t)=\sum_j j N_j(t)$, where  $N_i(t)=|A_i(t)|$.
The average of the initial opinions $c=S(0)/n$.}
Similarly,  
 let $Z(t)=n \sum_{v\in V} \pi_v X_v(t)$ be the total (degree biased) weight in the vertex process, and $c=Z(0)/n$ the initial (degree biased) average.
For regular graphs, $\pi_v=1/n$, in which case the processes coincide and $S(t)=Z(t)$.

Denote the total weight of the DIV process by $W(t)$. Thus $W(t)=S(t)$ in the edge process, and  $W(t)=Z(t)$ in the vertex process.
A random variable $W(t), t=0,1,...$   is a martingale with respect to a sequence $X(0),\cdots,X(t)$,
if its expected value at  step  $t+1$  depends only on  $X(t)$ and  satisfies $\E (W(t+1) \mid X(t)) = W(t)$.
The  next lemma shows that the  total weight $W(t)$ is a martingale with respect to the current vector of opinions $X(t)$.

\begin{lemma}\label{Lemma1}\label{SKn}\label{Lemma2}{\sc The average weight is a martingale.}

The following hold for each $t\ge 0$.
\vspace{-0.1in}
\begin{enumerate}[(i)]
\item
{\bf Asynchronous edge process.}   For arbitrary graphs, $S(t)$ is a martingale.
\item
{\bf Asynchronous vertex process. }  For arbitrary graphs, $Z(t)$ is a martingale.
\end{enumerate}
\end{lemma}

As the average opinion is a martingale (see Lemma \ref{Lemma1}), in cases where the process converges rapidly to two neighbouring states $\{i,i+1\}$, the 
Azuma-Hoeffding lemma (Lemma \ref{lem:Azuma-Hoeffding}) 
guarantees that the total weight is still asymptotic to its initial value.
Combining the above information with
 known results on the winning probabilities in two-opinion pull voting, see \eqref{probw},
 allows us   to predict the outcome of the process in Lemma \ref{ThA}(iii).

\begin{lemma}[The Azuma-Hoeffding inequality]
\label{lem:Azuma-Hoeffding}
Let $(X_t)_{t=0,1,2,...}$ be a martingale.
Suppose $|X_i-X_{i-1}|\leq d_i$ holds for any $i\geq 0$.
Then, for any $T\geq 0$ and $\epsilon>0$,
\begin{align*}
\Pr\left[|X_T-X_0|\geq \epsilon\right]\leq 2\exp\left(-\frac{\epsilon^2}{2\sum_{i=1}^Td_i^2}\right).
\end{align*}
\end{lemma}

\vspace{-0.075in}
For DIV, $d_i\le 1$ as opinions change by at most one at any step. Thus the total weight $W(t)$ satisfies
\begin{equation} \label{W-conc}
\Pr\left[|W(t)-W(0)|\geq h \right]\leq 2e^{-\frac{h^2}{2t}}. 
\end{equation}

As the process is randomized, the final value on a connected graph  is a  {random variable with} distribution $D(i)$ on the initial values $\{1,...,k\}$,
where $D(j)=\Pr(j \text{ wins})$.
The following lemma helps us to characterize this distribution in certain cases.
If only two consecutive opinions $i, i+1$ remain at some step $t$, the process is equivalent to two-opinion pull voting, and we say the incremental voting is at the final stage.

\begin{lemma}\label{ThA}{\sc Distribution of winning value.}\;\;
Let $W(t)=S(t)$ when referring to
the edge model, and let $W(t)=Z(t)$, when referring to the vertex model.
Let {$W(0)=cn$} be the total initial weight,
where $n$ is the number of vertices in the graph and $c$ is the initial average opinion.
\begin{enumerate}[(i)]
\item {
For an arbitrary graph,
the expected average opinion at any step is always the initial average: $\E[W(t)/n] = W(0)/n = c$.
The process $W(t)$ converges to a time invariant random variable.}
\item
For a connected graph,
if at the start of the final stage only two opinions $i$ and $i+1$ remain and the total weight $W$ is {$c'n$},
then for any connected graph,
the winning opinion is $i$ with probability $p=i+1-c'$, or $i+1$
with probability $q=c'-i$.
\item For a connected graph,
suppose the final stage is reached in $T$ steps,
where 
$T= o(n^2)$ for the asynchronous edge process,
and $T = o(1/\| \pi\|^2_\infty)$ for the asynchronous vertex process.
Then w.h.p. $W(T) \sim cn$ and the results of part (ii) hold with $c' \sim c$.
That is, for $i$ such that $i\le c < i+1$,
the winning opinion is $i$ with probability $p\sim i+1-c$, and is $i+1$
with probability $q\sim c-i$.
\end{enumerate}
\end{lemma}

\begin{lemma}\label{T2} {\sc Completion time, a general bound.}
For  asynchronous incremental voting on
connected graph, the worst-case
expected time to eliminate one of the two
extreme opinions (over all initial configurations) is at most the worst-case
expected completion time of standard asynchronous two-opinion $\{0,1\}$
voting.
\end{lemma}

\begin{corollary}\label{kl4jkd3w2}
The expected completion time of the discrete
incremental voting
is $O(k \cdot \cT_{2-vote})$,
where $\cT_{2-vote}$ is the worst-case\footnote{By worst case we mean with the worst initial mixture of 0's and 1's arranged on the vertices in
the worst manner.}  expected
completion time of the 2-opinion voting.
\end{corollary}

\section{Asynchronous DIV on expanders}
In this section, we prove our main theorem, \cref{thm:expander_gen},
which we restate here.
\begin{theorem}[Restatement of \cref{thm:expander_gen}]
\TheoremExpander
\end{theorem}
The proof goes in three stages. In the first part in Section \ref{reduce} we use the expander mixing lemma 
to estimate a time when at least one of the extreme opinions has small weight with positive probability. Essentially this depends
on exploiting the similarity of \eqref{avtrans} to the transition probability of a random walk.
In the second part in Section \ref{couple}
we couple the process with two-opinion voting; adapting a result from two-opinion voting  
which estimates a time when an opinion with small weight vanishes with positive probability.
In Section \ref{sec:Proof of main thm} we assemble these results to prove Theorem \ref{thm:expander_gen}.

\subsection{Reduction in size of the extreme opinions in DIV}\label{reduce}
Let $P$ be a transition matrix  of a simple random walk on $G$ and $\pi$ be its stationary distribution.
We assume $P$ is aperiodic, irreducible ($\forall v,u\in V$, $\exists t\geq 0$ s.t.~$P^t(v,u)>0$) and reversible ($\forall v,u\in V$, $\pi_vP(v,u)=\pi_uP(u,v)$).
Let $P(v,S)=\sum_{u\in S}P(v,u)$, $\pi(S)=\sum_{v\in S}\pi_v$ and $Q(U,S)=\sum_{v\in U}\pi_vP(v,S)$ for $S,U\subseteq V$.

Let $1=\lambda_1\geq \lambda_2\geq \cdots \geq \lambda_n$ be eigenvalues of $P$.
Let $\lambda=\max\{|\lambda_2|,|\lambda_n|\}$ be the second largest eigenvalue in absolute value.
Let $S^C=V\sm S$ denote the complement of a set of vertices $S$.
We use the following version of the expander mixing lemma.
See, e.g., the inequality below (12.9) in p.163 of~\cite{LP17}.
\begin{lemma}
\label{lem:EML}
Suppose $P$ is irreducible and reversible.
Then, for any $S,U\subseteq V$,
\begin{equation}\nonumber
    |Q(S,U)-\pi(S)\pi(U)|\leq \lambda\sqrt{\pi(S)\pi(S^C)\pi(U)\pi(U^C)}.
\end{equation}
\end{lemma}

The first step is to show that one of the extreme opinions gets sufficiently small with constant probability.
Consider  discrete incremental voting with opinions in  $[k]$.
Recall that
$X_v(t)\in [k]$ is the  opinion of $v$ at time $t$ and
$A_i(t)\subseteq V$ is the set of vertices holding the opinion $i\in [k]$.
Let
\[
s=\min_{v\in V}X_v(0) \textrm{ and } \ell=\min_{v\in V}X_v(0)
\]
be the smallest and largest opinions in the initial round.
Let
\begin{equation} \label{taudiv}
\tau^{\mathrm{DIV}}_{\mathrm{extr}}(\epsilon)=\min\{t\geq 0:\min\{\pi(A_s(t)),\pi(A_\ell(t))\}\leq \epsilon\}
\end{equation}
be the time when the weight of one of the extreme opinions gets smaller than $\epsilon$.
We present the following bounds for $\tau^{\mathrm{DIV}}_{\mathrm{extr}}(\epsilon)$ for
the two cases $\ell\geq s+3$ and $\ell= s+2$, which we treat separately.

The parameter $\eta$ in the next lemma is a failure probability in our calculations.
The specific values of $\eta$, $\epsilon_1$, and $\epsilon_2$ will be chosen in the proof of \cref{thm:expander_gen} in \cref{sec:Proof of main thm}.


\begin{lemma}
\label{lem:extremal_expander}
We have the following:
\begin{enumerate}[(i)]
    \item \label{item:extremalge4}
    Suppose $\ell\geq s+3$.
    Then, for any $\epsilon_1\geq 4\lambda^2$ and $\eta>0$, $\Pr[\tau^{\mathrm{DIV}}_{\mathrm{extr}}(\epsilon_1)>T_1]\leq \eta$ holds for
    $T_1=\left \lceil 2n\log\left(\frac{1}{4\epsilon_1^2\eta}\right)\right \rceil$. 
    \item \label{item:extremaleq3}
    Suppose $\ell=s+2$.
    Then, for any $\epsilon_2\geq 2\lambda$ and $\eta>0$,
    $\Pr[\tau^{\mathrm{DIV}}_{\mathrm{extr}}(\epsilon_2)>T_2]\leq \eta$
    holds for $T_2=\left\lceil \frac{2n}{\epsilon_2}\log \left(\frac{1}{4\epsilon_2^2\eta}\right)\right\rceil$.
\end{enumerate}
\end{lemma}
\begin{proof}
Write $A_i=A_i(t)$ and $A_i'=A_i(t+1)$ for convenience.
First, in both cases (that is, when $\ell \ge s+2$) from the definitions, we have
\begin{align*}
    \Pr[v\in A_s'\textrm{ and }u\in A_\ell']
    &=\begin{cases}
        1-\frac{2}{n}+\frac{P(v,A_s)}{n}+\frac{P(u,A_\ell)}{n} & (\textrm{if } v\in A_{s}\textrm{ and }u\in A_{\ell})\\
        \frac{P(v,A_s)}{n}& (\textrm{if } v\in A_{s+1}\textrm{ and }u\in A_{\ell})\\
        \frac{P(u,A_\ell)}{n}& (\textrm{if } v\in A_{s}\textrm{ and }u\in A_{\ell-1})\\
        0 & (\textrm{otherwise})
    \end{cases}.
\end{align*}
Hence, we have
\begin{align}
    &\E\left[\pi(A_{s}')\pi(A_\ell')\right] \nonumber\\
    &=\sum_{v\in V}\sum_{u\in V}\pi_v\pi_u\Pr\left[v\in A_s',u\in A_\ell'\right]\nonumber\\
    &=\pi(A_s)\pi(A_\ell)\left(1-\frac{2}{n}\right)+\pi(A_\ell)\frac{Q(A_s,A_s)}{n}+\pi(A_s)\frac{Q(A_\ell,A_\ell)}{n} \nonumber\\
    &\hspace{1em} +\pi(A_\ell)\frac{Q(A_{s+1},A_s)}{n}+\pi(A_s)\frac{Q(A_{\ell-1},A_\ell)}{n}\nonumber\\
    &=\pi(A_s)\pi(A_\ell)\left[1+\frac{1}{n}\left(\frac{Q(A_s,A_s\cup A_{s+1})}{\pi(A_s)}+\frac{Q(A_\ell,A_{\ell-1}\cup A_\ell)}{\pi(A_\ell)}-2\right)\right] \label{piAsAl}
\end{align}
Note that we used detailed balance, i.e., $Q(S,U)=Q(U,S)$ holds for any $S,U\subseteq V$,  to evaluate the summations in the above.

{\it Proof of (\ref{item:extremalge4}) of \cref{lem:extremal_expander}.}
From \cref{lem:EML}, we have for any $S,U\subseteq V$
\begin{align*}
    Q(S,U)
    &\leq \pi(S)\pi(U)+\lambda\sqrt{\pi(S)\pi(S^C)\pi(U)\pi(U^C)}
    \leq \pi(S)\pi(U)+\frac{\lambda \sqrt{\pi(S)}}{2}.
\end{align*}
Hence,
\begin{align}
    &\E\left[\pi(A_{s}')\pi(A_\ell')\right] \nonumber \\
    &\leq \pi(A_s)\pi(A_\ell)\left[1+\frac{1}{n}\left(\pi(A_s\cup A_{s+1})+\frac{\lambda}{2\sqrt{\pi(A_s)}}+\pi(A_{\ell-1}\cup A_\ell)+\frac{\lambda}{2\sqrt{\pi(A_\ell)}}-2\right)\right] \label{eq:asynckey1}\\
    &\leq \pi(A_s)\pi(A_\ell)\left[1+\frac{1}{n}\left(\frac{\lambda}{2\sqrt{\pi(A_s)}}+\frac{\lambda}{2\sqrt{\pi(A_\ell)}}-1\right)\right]. \label{eq:asyncexp1}
\end{align}
Note that $\pi(A_s\cup A_{s+1})+\pi(A_{\ell-1}\cup A_\ell)\leq 1$ holds since $s+1<\ell-1$.
Write $\tau=\tau^{\mathrm{DIV}}_{\mathrm{extr}}(\epsilon_1)$ for convenience.
Then,
\begin{align}
    &\mathbbm{1}_{\tau>t-1}\E\left[\pi(A_{s}(t))\pi(A_\ell(t))\mid X(t-1)\right] \nonumber \\
    &\leq \mathbbm{1}_{\tau>t-1}\pi(A_s(t-1))\pi(A_\ell(t-1))\left[1+\frac{1}{n}\left(\frac{\lambda}{2\sqrt{\pi(A_s(t-1))}}+\frac{\lambda}{2\sqrt{\pi(A_\ell(t-1))}}-1\right)\right] \nonumber \\
    &\leq \mathbbm{1}_{\tau>t-1}\pi(A_s(t-1))\pi(A_\ell(t-1))\left(1-\frac{1}{2n}\right) \label{eq:asynckey2}
\end{align}
holds for any $t\geq 1$.
Note that the definition of $\tau^{\mathrm{DIV}}_{\mathrm{extr}}(\epsilon_1)$ in \cref{taudiv} with $\eps_1\geq 4\l^2$ implies that both $\pi(A_{s}(t-1))\geq 4\lambda^2$ and $\pi(A_{\ell}(t-1))\geq 4\lambda^2$ for $\tau>t-1$.
Intuitively, \cref{eq:asynckey2} implies that $\pi(A_{s}(t))\pi(A_\ell(t))$ decreases by a factor of $1-1/2n$ at each time step until $\tau$ arrives.
To describe this intuition formally, let
\begin{align*}
    Y_t=\pi(A_s(t))\pi(A_\ell(t)), \;\;
    r=1-\frac{1}{2n}, \;\;
    Z_t=r^{-t}Y_t, \;\;
    \textrm{and } W_t=Z_{\tau\wedge t}.
\end{align*}
Then, from \cref{eq:asynckey2}, i.e., $\mathbbm{1}_{\tau>t-1}\E\left[Y_t\mid X(t-1)\right]\leq \mathbbm{1}_{\tau>t-1}\,rY_{t-1}$,
\begin{align*}
    \E\left[W_t-W_{t-1}\mid X(t-1)\right]
    &=\mathbbm{1}_{\tau>t-1}\,\E\left[Z_t-Z_{t-1}\mid X(t-1)\right]\\
    &=\mathbbm{1}_{\tau>t-1}\,r^{-t}\left(\E\left[Y_t\mid X(t-1)\right]-rY_{t-1}\right)\\
    &\leq 0
\end{align*}
holds. In other words, $(W_t)_{t=0,1,2,...}$ is a supermartingale.
Hence, we have
\begin{align}
    \E[W_T]\leq \E[W_0]=Y_0=\pi(A_s(0))\pi(A_\ell(0))\leq 1/4 \label{eq:asyncexp2}
\end{align}
and
\begin{align}
    \E[W_T]
    &\geq \E\left[W_T\mid \tau>T\right]\Pr[\tau>T]\nonumber\\
    &=\E\left[Z_T\mid \tau>T\right]\Pr[\tau>T]\nonumber\\
    &=r^{-T}\E\left[\pi(A_s(T))\pi(A_\ell(T))\mid \tau>T\right]\Pr[\tau>T]\nonumber\\
    &\geq r^{-T}\epsilon_1^2\Pr[\tau>T]. \label{eq:asyncexp3}
\end{align}
Note that $\pi(A_s(T))\geq \epsilon_1$ and $\pi(A_s(T))\geq \epsilon_1$ for  $\tau>T$.
Taking $T=\left \lceil 2n\log\left(\frac{1}{4\epsilon_1^2\eta}\right)\right \rceil$ $=O(n\log (\epsilon_1\eta)^{-1})$, \cref{eq:asyncexp2,eq:asyncexp3} give
\begin{align*}
    \Pr[\tau>T]\leq \frac{r^{T}}{4\epsilon_1^2}\leq \frac{1}{4\epsilon_1^2}\exp\left(-\frac{T}{2n}\right)\leq \eta.
\end{align*}

\noindent
{\it Proof of (\ref{item:extremaleq3}) of \cref{lem:extremal_expander}.}
Without loss of generality, we assume that $s=1$, i.e., $X_v(0)\in \{1,2,3\}$ holds for all $v\in V$.
    From \cref{lem:EML}, we have
$
    Q(S,U)
    \leq \pi(S)\pi(U)+\lambda\sqrt{\pi(S)\pi(U^C)},
$
so
$Q(A_1,A_1\cup A_2)/\pi(A_1)$ $\leq \pi(A_1\cup A_2)+\lambda\sqrt{\pi(A_3)/\pi(A_1)},
$
with an analogous bound for $Q(A_3,A_2\cup A_3)/\pi(A_3)$.
    Then continuing~\eqref{piAsAl} similarly as in~\eqref{eq:asynckey1}, we have
    \begin{align}
    &\E\left[\pi(A_{1}')\pi(A_3')\right] \nonumber \\
    &\leq \pi(A_1)\pi(A_3)\left[1+\frac{1}{n}\left(\pi(A_1\cup A_{2})+\lambda\sqrt{\frac{\pi(A_3)}{\pi(A_1)}}
    +\pi(A_{2}\cup A_3)+{\lambda}\sqrt{\frac{\pi(A_1)}{\pi(A_3)}}-2\right)\right] \nonumber\\
    &=\pi(A_1)\pi(A_3)\left[1+\frac{1}{n}\left(\sqrt{\pi(A_1)\pi(A_3)}\left(\frac{\lambda}{{\pi(A_1)}}+\frac{\lambda}{{\pi(A_3)}}\right)
    -\left(\pi(A_1)+\pi(A_3)\right)\right)\right]\label{eq:threeleft1}
    \end{align}
Note that if both $\pi(A_1)$ and $\pi(A_3)$ are at least $2\lambda$, then
\[
\sqrt{\pi(A_1)\pi(A_3)}\left(\frac{\lambda}{{\pi(A_1)}}+\frac{\lambda}{{\pi(A_3)}}\right)\leq \sqrt{\pi(A_1)\pi(A_3)}
\le \frac{\pi(A_1)+\pi(A_3)}{2}.
\]

Write $\tau=\tau^{\mathrm{DIV}}_{\mathrm{extr}}(\epsilon_2)$ for convenience.
    Since both $\pi(A_{1}(t-1))\geq \epsilon_2\geq 2\lambda$ and $\pi(A_{3}(t-1))\geq \epsilon_2\geq 2\lambda$ holds for $\tau>t-1$,
    \cref{eq:threeleft1} implies
    \begin{align}
        &\mathbbm{1}_{\tau>t-1}\E\left[\pi(A_{1}(t))\pi(A_3(t))\mid X(t-1)\right] \nonumber \\
        &\leq \mathbbm{1}_{\tau>t-1} \pi(A_{1}(t-1))\pi(A_3(t-1)) \left(1-\frac{\pi(A_1(t-1))+\pi(A_3(t-1))}{2n}\right) \nonumber \\
        &\leq \mathbbm{1}_{\tau>t-1} \pi(A_{1}(t-1))\pi(A_3(t-1)) \left(1-\frac{\epsilon_2}{2n}\right). \label{eq:threeleft2}
    \end{align}
    The last inequality follows from $\pi(A_{1}(t-1))\geq \epsilon_2$ and $\pi(A_{3}(t-1))\geq \epsilon_2$.
    Now, let
    \begin{align*}
        Y_t=\pi(A_{1}(t))\pi(A_3(t)), \,
        r=1-\frac{\epsilon_2}{2n}, \,
        Z_t=r^{-t}Y_t, \textrm{ and }
        W_t=Z_{\tau \wedge t}.
    \end{align*}
    From \cref{eq:threeleft2},
    \begin{align*}
        \E[W_t-W_{t-1}\mid X(t-1)]
        &=\mathbbm{1}_{\tau>t-1}\E[Z_t-Z_{t-1}\mid X(t-1)]\\
        &=\mathbbm{1}_{\tau>t-1}r^{-t}\left(\E[Y_t\mid X(t-1)]-rY_{t-1}\right) \\
        &\leq 0,
    \end{align*}
    i.e., $(W_t)_{t=0,1,2...}$ is a supermartingale.
    Hence, we have
    \begin{align}
        \E[W_T]\leq \E[W_0]=\pi(A_1(0))\pi(A_3(0))\leq 1/4. \label{eq:threeleftmar}
    \end{align}
    Furthermore,
    \begin{align}
        \E[W_T]&\geq \E[W_T\mid \tau>T]\Pr[\tau>T] \nonumber\\
        &=\E[Z_T\mid \tau>T]\Pr[\tau>T] \nonumber\\
        &=r^{-T}\E[\pi(A_1(T))\pi(A_3(T))\mid \tau>T]\Pr[\tau>T] \nonumber\\
        &\geq r^{-T}\epsilon_2^2\Pr[\tau>T] \label{eq:threeleft3}
    \end{align}
    holds. Combining \cref{eq:threeleftmar,eq:threeleft3} with $T=\left\lceil \frac{2n}{\epsilon_2}\log \left(\frac{1}{4\epsilon_2^2\eta}\right)\right\rceil$, 
    we obtain
    \begin{align*}
        \Pr[\tau>T] \leq \frac{r^{T}}{4\epsilon_2^2} \leq
        \frac{1}{4\epsilon_2^2}\exp\left(-\frac{\epsilon_2T}{2n}\right) \leq\eta.
    \end{align*}
\end{proof}

\subsection{Coupling DIV with pull voting}\label{couple}
Using \cref{lem:extremal_expander}, we  specify a time $T$, such that within $T$ steps, one of the extreme opinions disappears with a constant probability.
The main idea here is a coupling with pull voting.
Henceforth, we consider  two-opinion pull voting with opinions $\{1,2\}$. 
Let $B(t)\subseteq V$ be the set of vertices holding  opinion $1$.
Let
\[
\tau_{\mathrm{cons}}^{\mathrm{PULL}}=\min\{t\geq 0:B(t)=\emptyset  \textrm{ or } B(t)=V\}
\]
be the consensus time of two-opinion pull voting.

The next lemma, which is from \cite{CR16}, gives a probability bound on $\tau_{\mathrm{cons}}^{\mathrm{PULL}}$
in terms of the measure of the smallest opinion.
\begin{lemma}[\cite{CR16}]
\label{lem:pullupper}
$\Pr\left[\tau_{\mathrm{cons}}^{\mathrm{PULL}}>T\right]\leq 1/2$ for $T=\frac{64n}{\sqrt{2}(1-\lambda)\pi_{\min}}\sqrt{\min\{\pi(B(0)),\pi(B(0)^C)\}}$.
\end{lemma}
\begin{proof}
Let $B(0)=S$, and let
\begin{align*}
    \Psi&=\pi_{\min}\min_{S\subseteq V:S\neq \emptyset, V}\frac{\E[|\pi(B(1))-\pi(B(0))|\mid \! B(0)=S]}{\min\{\pi(S),1-\pi(S)\}},
\end{align*}
as in  \cite{CR16} expression (2).
It  follows from expression (15) in the proof of Theorem 2 of \cite{CR16} and the argument at the bottom of the same page that
\begin{align}\label{more-psi}
    \Pr\left[\tau_{\mathrm{cons}}^{\mathrm{PULL}}>\frac{64}{\sqrt{2}\Psi} \sqrt{\min\{\pi(B(0)),\pi(B(0)^C)\}}\right]\leq \frac{1}{2}.
\end{align}
See also Example 11  expressions (17), (18) of  \cite{CR16}  for more detail of the application of Theorem 2 to pull voting.

%

We next show that $\Psi\geq (\pi_{\min}(1-\lambda))/n$.
Consider the case where $v$ picks $u$.
Then, $|\pi(B(1))-\pi(B(0))|=\pi_v$ if $v$ and $u$ have different opinions,
and $|\pi(B(1))-\pi(B(0))|=0$ if $v$ and $u$ have the same opinion.
Hence, we have
\begin{align*}
\E[|\pi(B(1))-\pi(B(0))|\,\mid \!B(0)=S]
&=\sum_{v\in S}\pi_v\frac{1}{n}P(v,S^C)+\sum_{v\in S^C}\pi_v\frac{1}{n}P(v,S)\\
&=\frac 1n (Q(S,S^C)+Q(S^C,S))\\
&=\frac{2Q(S,S^C)}{n}\\
&\geq \frac{2(1-\lambda)\pi(S)\pi(S^C)}{n}.
\end{align*}

In the above we used detailed balance to obtain $Q(S^C,S)=Q(S,S^C)$,
and in the last inequality, we used $Q(S,S^C)\geq (1-\lambda)\pi(S)\pi(S^C)$ which comes from \cref{lem:EML}.
Furthermore, since 
$\max\{\pi(S),\pi(S^C)\}\geq 1/2$,
\[
\pi(S)\pi(S^C)=\min\{\pi(S),\pi(S^C)\}\max\{\pi(S),\pi(S^C)\}\geq \min\{\pi(S),\pi(S^C)\}/2.
\]
Thus
\[
 \frac{\E[|\pi(B(1))-\pi(B(0))|\,\mid \!B(0)=S]}{\min\{\pi(S),\pi(S^C)\}}
\geq  \frac{(1-\lambda)\min\{\pi(S),\pi(S^C)\}}{n\min\{\pi(S),\pi(S^C)\}}=\frac{(1-\lambda)}{n}
\]
holds for any $S\subseteq V$ if $S\neq \emptyset, V$. 

It follows that $\Psi \ge \pi_{\min} (1-\lambda)/n$. 
Using \eqref{more-psi},  we obtain the claim.
\end{proof}

Recall the definition of  $\tau^{\mathrm{DIV}}_{\mathrm{extr}}(0)$  as given in \eqref{taudiv}.
We next give an estimate in terms of $\tau_{\mathrm{cons}}^{\mathrm{PULL}}$, of the time $\tau^{\mathrm{DIV}}_{\mathrm{extr}}(0)$  at which one of the extreme opinions disappears in  discrete incremental voting.
\begin{lemma}
\label{lem:extreme_disapper_exp}
Suppose $\Pr\left[\tau_{\mathrm{cons}}^{\mathrm{PULL}}>T_p\sqrt{\min\{\pi(B(0)),\pi(B(0)^C)\}}\right]\leq 1/2$ holds for some $T_p$.
Then, we have $\Pr\left[\tau^{\mathrm{DIV}}_{\mathrm{extr}}(0)>T_p\sqrt{\min\{\pi(A_s(0)),\pi(A_\ell(0))\}}\right]\leq 1/2$.
%
%
\end{lemma}
We will choose $T_p=\frac{64n}{\sqrt{2}(1-\lambda)\pi_{\min}}$ as in \cref{lem:pullupper}.
To show \cref{lem:extreme_disapper_exp}, we use the next lemma, which is a rephrasing of Lemma \ref{T2}.

\begin{lemma}
\label{lem:domination}
Let $B(0)$ be the set of vertices initially holding opinion 1 in two-opinion pull voting with opinions $\{1,2\}$.
We have the following:
\begin{enumerate}[(i)]
    \item Suppose $A_s(0)= B(0)$. Then, there is a coupling such that both $A_s(t)\subseteq B(t)$ and $A_\ell(t)\subseteq V\setminus B(t)$ hold for any $t\geq 0$.
    \item Suppose $A_\ell(0)= B(0)$. Then, there is a coupling such that both $A_\ell(t)\subseteq B(t)$ and $A_s(t)\subseteq V\setminus B(t)$ hold for any $t\geq 0$.
\end{enumerate}
\end{lemma}

\begin{proof}[{\bf Proof of \cref{lem:extreme_disapper_exp}}]
    First, consider the case of $\pi(A_s(0))\leq \pi(A_\ell(0))$.
    In this case, \\
     since $\pi(A_s(0))\leq \pi(A_\ell(0))\leq 1-\pi(A_s(0))$, $\min\{\pi(A_s(0)),\pi(A_s(0)^C)\}=\pi(A_s(0))$.
    Let
    \begin{align*}
        T_s=T_p\sqrt{\min\{\pi(A_s(0)),\pi(A_s(0)^C)\}}=T_p\sqrt{\pi(A_s(0))}=T_p\sqrt{\min\{\pi(A_s(0)),\pi(A_\ell(0))\}}.
    \end{align*}
    Then, applying \cref{lem:domination} with $B(0)=A_s(0)$,
    \begin{align*}
        \Pr[\tau^{\mathrm{DIV}}_{\mathrm{extr}}(0)>T_p\sqrt{\min\{\pi(A_s(0)),\pi(A_\ell(0))\}}]
        &= \Pr[\tau^{\mathrm{DIV}}_{\mathrm{extr}}(0)>T_s]\\
        &= \Pr[A_s(T_s)\neq \emptyset \textrm{ and } A_\ell(T_s) \neq \emptyset]\\
        &\leq \Pr[B(T_s)\neq \emptyset  \textrm{ and } V\setminus B(T_s) \neq \emptyset ]\\
        &\leq \Pr[\tau^{\mathrm{PULL}}_{\mathrm{cons}}> T_s]\\
        &\leq 1/2.
    \end{align*}

    Second, suppose $\pi(A_s(0))>\pi(A_\ell(0))$.
    In this case, $\pi(A_\ell(0))< \pi(A_s(0))\leq 1-\pi(A_\ell(0))$.
    Let
    \begin{align*}
        T_\ell=T_p\sqrt{\min\{\pi(A_\ell(0)),\pi(A_\ell(0)^C)\}}=T_p\sqrt{\pi(A_\ell(0))}=T_p\sqrt{\min\{\pi(A_s(0)),\pi(A_\ell(0))\}}.
    \end{align*}
    Then, applying \cref{lem:domination} with $B(0)=A_\ell(0)$,
    \begin{align*}
        \Pr[\tau^{\mathrm{DIV}}_{\mathrm{extr}}(0)>T_p\sqrt{\min\{\pi(A_s(0)),\pi(A_\ell(0))\}}]
        &= \Pr[\tau^{\mathrm{DIV}}_{\mathrm{extr}}(0)>T_\ell] \\
        &= \Pr[A_\ell(T_\ell)\neq \emptyset \textrm{ and } A_s(T_\ell) \neq \emptyset]\\
        &\leq \Pr[B(T_\ell)\neq \emptyset  \textrm{ and } V\setminus B(T_\ell) \neq \emptyset ]\\
        &\leq \Pr[\tau^{\mathrm{PULL}}_{\mathrm{cons}}> T_\ell]\\
        &\leq 1/2
    \end{align*}
    holds. 
    Thus, we obtain the claim.
\end{proof}

\subsection{Proof of the main theorem, Theorem \ref{Expandy}}
\label{sec:Proof of main thm}
In this section we prove \cref{Expandy}  by combining \cref{lem:extremal_expander,lem:extreme_disapper_exp} in Lemma \ref{lem:general}.
The parameters $T_1(\epsilon), T_2(\epsilon)$, and $T_p(\epsilon)$ from those lemmas are as follows:
\begin{flalign}\label{Tvals}\!\!
    T_1(\epsilon)=\left \lceil 2n\log\left(\frac{1}{2\epsilon^2}\right)\right \rceil,
    T_2(\epsilon)=\left\lceil \frac{2n}{\epsilon}\log \left(\frac{1}{2\epsilon^2}\right)\right\rceil,
     T_p=\left\lceil\frac{64n}{\sqrt{2}(1-\lambda)\pi_{\min}}\right\rceil.
\end{flalign}
For convenience, let $T_p(\epsilon):=T_p\sqrt{\epsilon}$.

\begin{lemma}
\label{lem:general}
We have the following:
\begin{enumerate}[(i)]
    \item \label{item:disapper_general_ge4}
    Suppose $\ell\geq s+3$.
    Then, for any $\epsilon_1\geq 4\lambda^2$,
    $\Pr\left[\tau^{\mathrm{DIV}}_{\mathrm{extr}}(0)>T_1(\epsilon_1)+T_p(\epsilon_1)\right]\leq 3/4$
    and
    $\E\left[\tau^{\mathrm{DIV}}_{\mathrm{extr}}(0)\right]
    \leq 4(T_1(\epsilon_1)+T_p(\epsilon_1))$
    hold.
    \item \label{item:disapper_general_eq3}
    Suppose $\ell=s+2$.
    Then, for any $\epsilon_2\geq 2\lambda$,
    $\Pr\left[\tau^{\mathrm{DIV}}_{\mathrm{extr}}(0)>T_2(\epsilon_2)+T_p(\epsilon_2)\right]\leq 3/4$  and
    $\E\left[\tau^{\mathrm{DIV}}_{\mathrm{extr}}(0)\right]
    \leq 4(T_2(\epsilon_2)+T_p(\epsilon_2))$
    hold.
\end{enumerate}
\end{lemma}
\begin{proof}[Proof of (\ref{item:disapper_general_ge4}) of \cref{lem:general}]
Applying \cref{lem:extremal_expander}(i) with $\eta=1/2$,
it holds that within $t\leq T_1(\epsilon_1)$ steps,
$\min\{\pi(A_s(t)),\pi(A_\ell(t))\}\leq \epsilon_1$ with a probability at least $1/2$.
Furthermore, starting with an initial condition where $\min\{\pi(A_s(0)),\pi(A_\ell(0))\}\leq \epsilon_1$,
\cref{lem:pullupper,lem:extreme_disapper_exp} imply that
within $t\leq T_p(\epsilon_1)$ steps, $\min\{\pi(A_s(t)),\pi(A_\ell(t))\}=0$ with a probability at least $1/2$.
Combining these facts, we have
\[
    \Pr\left[\tau^{\mathrm{DIV}}_{\mathrm{extr}}(0)\leq T_1(\epsilon_1)+T_p(\epsilon_1)\right]\geq \left(\frac{1}{2}\right)^2=\frac{1}{4}.
\]
Repeating these arguments implies that $\Pr\left[\tau^{\mathrm{DIV}}_{\mathrm{extr}}(0)> h\,(T_1(\epsilon_1)+T_p(\epsilon_1))\right]\leq (3/4)^h$.
Hence,
\begin{align*}
    \E\left[\tau^{\mathrm{DIV}}_{\mathrm{extr}}(0)\right]
    &=\sum_{t=0}^\infty \Pr\left[\tau^{\mathrm{DIV}}_{\mathrm{extr}}(0)>t\right]\\
    &=\sum_{h=0}^\infty\sum_{j=0}^{T_1(\epsilon_1)+T_p(\epsilon_1)-1}\Pr\left[\tau^{\mathrm{DIV}}_{\mathrm{extr}}(0)>h(T_1(\epsilon_1)+T_p(\epsilon_1))+j\right]\\
    &\leq \sum_{h=0}^\infty (T_1(\epsilon_1)+T_p(\epsilon_1)) \left(\frac{3}{4}\right)^h\\
    &=4(T_1(\epsilon_1)+T_p(\epsilon_1)).
\end{align*}
\end{proof}

\begin{proof}[Proof of (\ref{item:disapper_general_eq3}) of \cref{lem:general}]
Applying \cref{lem:extremal_expander}(ii) with $\eta=1/2$,
it holds that within $t\leq T_2(\epsilon_2)$ steps,
$\min\{\pi(A_s(t)),\pi(A_\ell(t))\}\leq \epsilon_2$ with a probability at least $1/2$.
Applying \cref{lem:pullupper,lem:extreme_disapper_exp}, with the initial condition of $\min\{\pi(A_{s}(0)),\pi(A_{\ell}(0))\}=\epsilon_2$,
it holds that within $t\leq T_p(\epsilon_2)$ steps, $\min\{\pi(A_{s}(t)),\pi(A_{\ell}(t))\}=0$ with a probability at least $1/2$.
These facts imply that
\[
    \Pr\left[\tau^{\mathrm{DIV}}_{\mathrm{extr}}(0)\leq T_2(\epsilon_2)+T_p(\epsilon_2)\right]\geq \left(\frac{1}{2}\right)^2=\frac{1}{4}.
\]
Hence, we obtain $\Pr\left[\tau^{\mathrm{DIV}}_{\mathrm{extr}}(0)> h( T_2(\epsilon_2)+T_p(\epsilon_2))\right]\leq (3/4)^h$
for any integer $h>0$ and
\begin{align*}
    \E\left[\tau^{\mathrm{DIV}}_{\mathrm{extr}}(0)\right]
    &=\sum_{t=0}^\infty \Pr\left[\tau^{\mathrm{DIV}}_{\mathrm{extr}}(0)>t\right]\\
    &=\sum_{h=0}^\infty\sum_{j=0}^{T_2(\epsilon_2)+T_p(\epsilon_2)-1}\Pr\left[\tau^{\mathrm{DIV}}_{\mathrm{extr}}(0)>h(T_2(\epsilon_2)+T_p(\epsilon_2))+j\right]\\
    &\leq \sum_{h=0}^\infty (T_2(\epsilon_2)+T_p(\epsilon_2)) \left(\frac{3}{4}\right)^h\\
    &=4(T_2(\epsilon_2)+T_p(\epsilon_2)).
\end{align*}
\end{proof}

\ignore{
\begin{proof}[Proof of (\ref{super-conc}) of \cref{lem:general}]
{\blue Let $J=1,2$ and  $T=T_J(\epsilon_J)+T_p(\epsilon_J)$. As $\Pr(\tau^{\mathrm{DIV}}_{\mathrm{extr}}(0)\le T)\ge 1/4$, restarting the process every $T$ steps,  the probability of success is independent. 
Repeating this argument $h$ times, the probability that   neither the minimum or maximum opinion has disappeared
within $hT$ steps is at most $4(3/4)^h$. As $\E[T] \le 4T$, the claim follows.}
\end{proof}

{\violet  
From \cref{lem:general},
\begin{align}
    \E[\tau]\leq 4(k-3)(T_1(\epsilon_1)+T_p(\epsilon_1))+4(T_2(\epsilon_2)+T_p(\epsilon_2)),
    \label{eq:expected value bound}
\end{align}
{\red which by Lemma \ref{lem:extremal_expander}} holds for any $\epsilon_1\geq 4\lambda^2$ and $\epsilon_2\geq 2\lambda$.
By Lemma \cref{lem:general}(\ref{item:disapper_general_eq3})  we have that
\[
\Pr(\tau \ge h \E[\tau]) =O\brac{ k \, (3/4)^{h/4}}.
\]
}
}

\begin{proof}[{\bf Proof of \cref{thm:expander_gen}}]
Write $\tau$ for the first time $t\geq 0$ when there are at most two consecutive opinions $i$ and $i+1$ remaining.
Henceforth, we abbreviate $\max_{x}\E[\tau|X(0)=x]$ (the expected value of $\tau$ from the worst initial configuration) as $\E[\tau]$ for convenience.
From \cref{lem:general},
\begin{align}
    \E[\tau]\leq 4(k-3)(T_1(\epsilon_1)+T_p(\epsilon_1))+4(T_2(\epsilon_2)+T_p(\epsilon_2)),
    \label{eq:expected value bound}
\end{align}
which by Lemma \ref{lem:extremal_expander} holds for any $\epsilon_1\geq 4\lambda^2$ and $\epsilon_2\geq 2\lambda$.

Now, set
$\epsilon_1=\max\{4\lambda^2,n^{-2}\}\geq 4\lambda^2$ and
$\epsilon_2=\max\{2\lambda,n^{-2/3}\}\geq 2\lambda$.
Assumptions of $\pi_{\min}=\Theta(1/n)$ and $\lambda =o(1)$ imply that $\frac{n}{(1-\lambda)\pi_{\min}}=O(n^2)$.
Hence, from \cref{Tvals} there is a sufficiently large constant $C>0$ such that
\begin{align*}
    &kT_1(\epsilon_1)\leq Cnk\log n,
    &&kT_p(\epsilon_1)\leq Ckn^2\max\{\lambda, 1/n\}\leq C(n^2\lambda k+nk), \\
    &T_2(\epsilon_2)\leq Cn^{5/3}\log n,
    &&T_p(\epsilon_2)\leq Cn^2\max\{\sqrt{\lambda}, 1/n^{1/3}\}\leq C(n^2\sqrt{\lambda}+n^{5/3})
\end{align*}
hold. Hence, noting the assumptions that $k=o(n/\log n)$ and $\lambda k=o(1)$, it follows that all of $kT_1(\epsilon_1)$, $kT_p(\epsilon_1)$, $T_2(\epsilon_2)$, and $T_p(\epsilon_2)$ are $o(n^2)$. 
Putting this together gives the following upper bound on \Cref{eq:expected value bound},
\begin{align}
    \E[\tau]=O(nk \log n +n^2 \lambda k +n^{5/3} \log n +n^2 \sqrt{\lambda}).
    \label{eq:tau upper bound}
\end{align}

\ignore{
(We can delete this blue part)
Thus provided $k=o(n \log n)$ and $\lambda k =o(1)$,
\Cref{eq:tau upper bound} implies that $\E[\tau]=o(n^2)$., i.e., there is some function $\epsilon(n)=o(1)$ such that
$\E[\tau]\leq n^2\epsilon(n)$ holds. 

From the Markov inequality, we obtain
\begin{align}
    \Pr\left[\tau< n^2\sqrt{\epsilon(n)}\right]
    &=1-\Pr\left[\tau\geq  n^2\sqrt{\epsilon(n)}\right]
    \geq 1-\frac{\E[\tau]}{n^2\sqrt{\epsilon(n)}}
    \geq 1-\sqrt{\epsilon(n)}
    =1-o(1) .
    \label{eq:o(n^2) bound}
\end{align}
}

Provided $k=o(n \log n)$ and $\lambda k =o(1)$, from \eqref{eq:expected value bound}
we obtain $\E[\tau]=o(n^2)$, i.e., there is some function $\epsilon(n)=o(1)$ such that
$\E[\tau]\leq n^2\epsilon(n)$ holds. 

The Markov inequality, $\Pr[\tau>\mathrm{e}\E[\tau]]\leq 1/\mathrm{e}$ holds for any initial configuration.
By repeating this process independently $h$ times, we obtain the following consequence:
\begin{align}\label{eq:exponential decay}
    \Pr[\tau>h\mathrm{e}\E[\tau]]\leq \mathrm{e}^{-h}.
\end{align}
Thus, taking $h= \left\lceil 1/\sqrt{\epsilon(n)} \right\rceil=\omega(1)$, we obtain $\Pr[\tau\leq h\mathrm{e}\E[\tau]]\geq 1-\mathrm{e}^{-\omega(1)}$ for $h\mathrm{e} \E[\tau]\leq \mathrm{e} n^2\epsilon(n)\left(\frac{1}{\sqrt{\epsilon(n)}}+1\right)=o(n^2)$.

In other words, within $o(n^2)$ steps,
there are at most two consecutive opinions remaining w.h.p.
This completes the proof of \cref{thm:expander_gen}.
\end{proof}

\newpage


\begin{thebibliography}{99}
\bibitem{AD}  Mohammed Amin Abdullah and Moez Draief.
Global majority consensus by local majority
polling on graphs of a given degree sequence. Discrete Applied Mathematics, 180 1--10, 2015.



\bibitem{Becchetti} Luca Becchetti, Andrea Clementi, Emanuele Natale, Francesco Pasquale, Riccardo Silvestri, and
Luca Trevisan. Simple dynamics for plurality consensus.
SPAA 2014,  247–256, New York, NY,
USA, 2014. ACM.

{%
\bibitem{Bec2} Luca Becchetti, Andrea E. F. Clementi, Emanuele Natale, Francesco Pasquale,
Riccardo Silvestri: Plurality Consensus in the Gossip Model. SODA 2015:
371-390

\bibitem{Bec3} Luca Becchetti, Andrea E. F. Clementi, Emanuele Natale, Francesco Pasquale,
Luca Trevisan: Stabilizing Consensus with Many Opinions. SODA 2016: 620-635
}

\bibitem{Petra2} Petra Berenbrink, Tom Friedetzky, Dominik Kaaser, and Peter Kling. 
Tight \& Simple Load Balancing. 
In 2019 IEEE International Parallel and Distributed Processing Symposium (IPDPS). 718--26.

\bibitem{Petra} Petra Berenbrink, George Giakkoupis, and Peter Kling. 
Tight bounds for coalescing-branching random walks on regular graphs. 
SODA 2018, 1715–1733 Society for Industrial and Applied Mathematics, Philadelphia, PA, USA, 2018.





\bibitem{BMPS04}
Siddhartha Brahma, Sandeep Macharla, Sudebkumar Prasant Pal, and Sudhir Kumar Singh. 
{Fair leader election by randomized voting.}
{ICDCIT 2004},
 22--31, (2004).

\bibitem{ACO}
Amin Coja-Oghlan.
On the Laplacian Eigenvalues of $G_{n,p}$. 
Combinatorics, Probability and Computing, 16, pages 923--946 (2007).

\bibitem{CGJ18}
Nicholas Cook, Larry Goldstein, and Tobias Johnson.
Size biased couplings and the spectral gap for random regular graphs.
The Annals of Probability 46 (1): 72–125 (2018).


\bibitem{CER} 
Colin Cooper, Robert Elsässer, and Tomasz Radzik. 
The power of two choices in distributed voting. ICALP 2014,
435–446, Berlin, Heidelberg, 2014. Springer Berlin Heidelberg.

\bibitem{CERReigen}
Colin Cooper, Robert  Els\"asser, Tomasz Radzik, Nicol\'as Rivera, and Takeharu Shiraga. 
Fast consensus for voting on general expander
graphs. DISC 2015 248-262 (2015).


\bibitem{CRRS} 
Colin Cooper, Tomasz Radzik, Nicol\'as Rivera, and Takeharu Shiraga. 
Fast plurality consensus in regular expanders.  DISC 2017 1–13:16 (2017).


\bibitem{DIVFULL} 
Colin Cooper, Tomasz Radzik, and Takeharu Shiraga. 
Discrete incremental voting. 
OPODIS 2023 10:1-10:22, and  arXiv Technical Report 2305.15632, (2023).


\bibitem{CR16} 
Colin Cooper and Nicol\'as Rivera. 
The linear voting model. 
In 43rd International Colloquium on Automata, Languages, and Programming (ICALP 2016), pages 144:1–144:12 (2016).

\bibitem{Doerr}
Benjamin Doerr, Leslie Ann Goldberg, Lorenz Minder, Thomas Sauerwald, and Christian Scheideler.
Stabilizing consensus with the power of two choices. 
In Proc. of the 23rd Ann. ACM Symp. on Parallelism in Algorithms and Architectures (SPAA’11), 149–158. ACM, 2011.

\bibitem{Ghaffari}
Mohsen Ghaffari and Johannes Lengler. 
Nearly-tight analysis for 2-choice and 3-majority consensus dynamics. 
In Proceedings of the 2018 ACM Symposium on Principles of Distributed Computing, PODC ’18, 305–313, New York, NY, USA, 2018. ACM.

\bibitem{HassinPeleg-InfComp2001}
Yehuda Hassin and David Peleg.
Distributed probabilistic polling and applications to proportionate agreement.
Information \& Computation, 171, pages 248--268, (2001).

\bibitem{Joh89} 
Barry W.~Johnson, 
{\em Design and Analysis of Fault Tolerant Digital Systems},
Addison-Wesley, (1989).



\bibitem{NanNico}
Nan Kang, Nicolas Rivera.
Best-of-Three Voting on Dense Graphs.
In Proceedings of the 31st ACM Symposium on Parallelism in Algorithms and Architectures (SPAA), 115–121, 2019.


\bibitem{LP17} 
David A Levin and Yuval Peres. 
Markov Chains and Mixing Times. 
American Mathematical Society, 2017.

\bibitem{Nakata_etal_1999}
Toshio Nakata, Hiroshi Imahayashi, and Masafumi Yamashita.
Probabilistic local majority voting for the agreement problem on finite graphs.
{ COCOON 1999},
330--338, (1999).

\bibitem{NS}
Nobutaka Shimizu, Takeharu Shiraga.
Phase transitions of Best-of-two and Best-of-three on stochastic block models. 
Random Struct. Algorithms 59(1): 96-140 (2021)



\bibitem{TY19}
Konstantin Tikhomirov and Pierre Youssef.
The spectral gap of dense random regular graphs.
The Annals of Probability, 47 (1): 362–419 (2019).

\end{thebibliography}


\end{document}